\newif\iftwocolumn
\twocolumnfalse 

\newif\ifhyper
\hyperfalse 

\iftwocolumn

\def\figsizeb{\linewidth}
\else

\def\figsizeb{0.55\linewidth}
\def\figsizec{0.65\linewidth}
\fi

\iftwocolumn
\documentclass[prl,aps,twocolumn,showpacs]{revtex4}
\else
\documentclass[12pt]{iopart}
\fi
\usepackage[dvipdfm]{graphicx}
\usepackage{dcolumn}
\usepackage{bm}
\expandafter\let\csname equation*\endcsname\relax 
\expandafter\let\csname endequation*\endcsname\relax 
\usepackage{amsmath}
\usepackage{amsthm}
\usepackage{amssymb}
\usepackage{cite}

\newtheorem{thm}{Theorem}

\ifhyper
\usepackage{atbegshi}
  \AtBeginShipoutFirst{}
\usepackage[colorlinks=true,linkcolor=blue,citecolor=blue,anchorcolor=blue,urlcolor=blue,pdfborderstyle={},dvipdfm]{hyperref}
\fi

\begin{document}

\title[Multifractals Competing with Solitons on Fibonacci Optical Lattice]{Multifractals Competing with Solitons on Fibonacci Optical Lattice}

\author{M.~Takahashi$^1$, H.~Katsura$^1$, M.~Kohmoto$^2$, and T.~Koma$^1$}

%
\address{$^1$Department of Physics, Gakushuin University, Tokyo 171-8588, Japan}

\address{$^2$Institute for Solid State Physics, University of Tokyo, Kashiwa, Chiba 277-8581, Japan}

\eads{\mailto{masahiro.takahashi@gakushuin.ac.jp} \\ \mailto{hosho.katsura@gakushuin.ac.jp} \\
\mailto{kohmoto@issp.u-tokyo.ac.jp} \\ 
\mailto{tohru.koma@gakushuin.ac.jp}}

\date{\today}

\begin{abstract}
We study the stationary states for the nonlinear Schr\"odinger equation on the Fibonacci lattice 
which is expected to be realized by Bose-Einstein condensates loaded into an optical lattice. 
When the model does not have a nonlinear term, the wavefunctions and the spectrum are known 
to show fractal structures. Such wavefunctions are called critical. 
We present a phase diagram of the energy spectrum for varying the nonlinearity. 
It consists of three portions, 
a forbidden region, the spectrum of critical states, and the spectrum of stationary solitons.
We show that the energy spectrum of critical states remains intact 
irrespective of the nonlinearity in the sea of a large number of stationary solitons.
\end{abstract}

\submitto{\NJP}
\maketitle

\section{Introduction}
The realization of Bose-Einstein condensation (BEC) in optical lattices has opened a new avenue 
for studying a variety of phenomena in condensed matter systems~\cite{Bloch_NP_2005_v0, Morsch_RMP_2006_v0}. 
A major advantage of using ultracold atomic gases 
is that one can control the interatomic interactions and the lattice parameters 
in an extremely clean environment. 
This high degree of tunability enables one to study BEC in artificially designed structures 
which cannot be achieved in conventional solids. 
For instance, the one-dimensional bichromatic potential 
has been realized in a system of $^{87}$Rb atoms~\cite{Fallani_PRL_2007_v0}, 
and such a quasiperiodic potential has been studied theoretically~\cite{Larcher_PRA_2009_v0} 
in the context of BEC. 
More recently, a new method for creating potentials through a holographic mask 
was introduced~\cite{Bakr_Nature_2009_v0}. 
Using this technique, it seems feasible to experimentally generate exotic structures
such as the Fibonacci lattice~\cite{Kohmoto_PRL_1983_v0, Ostlund_PRL_1983_v0, Kohmoto_PRB_1987_v0} 
and the Penrose tiling~\cite{Penrose_Bull.-Inst.-Math.-Appl._1974_v0, Gardner_Sci.-Am._1977_v0} 
which are of interest as one- and two-dimensional quasicrystals, respectively.

These low-dimensional quasicrystals have attracted considerable theoretical attention 
since there appear fractal wavefunctions  which are neither extended nor localized 
\cite{Kohmoto_PRB_1987_v0, Suto_JSP_1989_v0}. 
They are called critical states. The rapid progress in the study of BEC can set  
a stage for exploring these exotic states in experiments. 
The new ingredient that appears in the system of BEC is the nonlinearity caused 
by interatomic interactions which can be tuned by the Feshbach resonance. 
One might expect that 
fractal wavefunctions are fragile and are easily destroyed by the nonlinearity. 
Surprisingly, this is not the case. 
In fact, we demonstrate that there indeed exist critical states on the Fibonacci optical lattice. 
This is intended to stimulate experimental efforts to observe critical states 
in a cold-atom setup.\footnote{
Quite recently, disorder effects were studied for 
transport in photonic quasicrystals \cite{Levi_Science_2011_v0}. 
Their experimental results show that certain disorder enhances the transport.
}

In order to describe a Bose-Einstein condensate on the Fibonacci optical lattice, 
we resort to the nonlinear Schr\"odinger equation with on-site potentials arranged 
in the Fibonacci sequence~\cite{Johansson_PRB_1994_v0}.
In the absence of the nonlinear term, 
it is known that all the eigenstates are critical, 
and that the spectrum shows a fractal structure~\cite{Kohmoto_PRB_1987_v0, Suto_JSP_1989_v0}. 
More precisely, the spectrum is 
singular continuous~\cite{Suto_JSP_1989_v0,Kohmoto_Physics-Letters-A_1984_v0} 
and called the Cantor spectrum. 
In order to elucidate the effect of nonlinearity on 
the critical states, 
we numerically solve the stationary nonlinear Schr\"odinger equation. 
As mentioned above, our numerical results show that 
the critical states persist despite the presence of the nonlinearity in the sea of  
stationary solitons~\cite{Kivshar__2003_v0, Johansson_PRB_1994_v0}. 
With the aid of mathematical tools, we show that for any critical state, 
the ``eigenenergy" must be included in the Cantor spectrum of the 
Schr\"odinger equation without nonlinearity. 
Further, we determine a forbidden region for the eigenenergy and the strength 
of the nonlinearity. Putting these together, we present a phase diagram of 
the energy spectrum for varying the nonlinearity. (See Fig.~\ref{fig:phase} 
in Sec.~\ref{PhaseDiagram}.)  
The energy spectrum of the critical states retains its profile 
irrespective of the nonlinearity, 
while the number of the stationary solitons increases enormously 
as the nonlinearity increases. 
One might think that the presence of the sea of solitons  
makes it difficult to experimentally detect the fractal profiles of the critical states. 
However, in the neighborhood of the forbidden region, 
such an experimental detection is expected to be possible.  
(See Sec.~\ref{PhaseDiagram} and Fig.~\ref{fig:phase} for details.) 

Throughout the present paper, we will not treat dynamical properties of the wavefunctions 
for the nonlinear Schr\"odinger equation. However, we think we should at least stress that 
knowledge of the stationary states is not sufficient for understanding the dynamics
such as diffusion of wave packets in quasiperiodic or 
random environment~\cite{Pikovsky_PRL_2008_v0, Wang_JSP_2009_v0, Flach_PRL_2009_v0, 
Skokos_PRE_2009_v0, Fishman_Nonlinearity_2009_v0, Krivolapov_NJP_2010_v0}. 
For the random nonlinear Schr\"odinger equation, 
see a recent review~\cite{Fishman_Nonlinearity_2012_v0}.  

The present paper is organized as follows: The precise definition 
of the nonlinear Schr\"odinger equation 
which we consider in the present paper is given in Sec.~\ref{sec:Preliminary}. 
In Sec.~\ref{sec:NA}, the numerical solutions of the model are obtained. 
In Sec.~\ref{sec:MulA}, we apply the multifractal 
analysis to the solutions so obtained. The mathematical analysis for 
the model is given in Sec.~\ref{sec:MaA}. Our results are summarized as 
the phase diagram in Sec.~\ref{PhaseDiagram}. 
Section~\ref{SumCon} is devoted to summary and conclusion. 
In \ref{AppendixA}, we discuss effects of nonlinearity on localization.

\section{Preliminary} \label{sec:Preliminary}
Let us consider the one-dimensional chain which is determined
by the Fibonacci rule. The $\ell$-th chain, $S_\ell$, $\ell \! = \! 1, 2, \ldots$, consists of
two symbols, $A$ and $B$, and is constructed by the recursion,
$S_{\ell \! + \! 1} \! = \! S_\ell S_{\ell \! - \! 1}$, with the initial condition, $S_0 \! = \! B$ 
and $S_1 \! = \! A$:
\begin{align}
S_1\!=\!A, ~~ S_2\!=\!AB, ~~ S_3\!=\!ABA, ~~ S_4\!=\!ABAAB, ~ \ldots \nonumber
\end{align}
We denote by $N \! = \! N_\ell$ the number of symbols in $S_\ell$. Clearly,
$N_\ell$ is equal to the Fibonacci number because they satisfy
$N_{\ell \! + \! 1} \! = \! N_\ell \! + \! N_{\ell \! - \! 1}$ with $N_0 \! = \! N_1 \! = \! 1$. 
In the limit $\ell \! \rightarrow \! \infty$, the Fibonacci chain is neither random nor periodic. 
In fact, it is \textit{quasiperiodic}.

The nonlinear Schr\"odinger equation for stationary states,
$(\psi_1,\psi_2,\ldots,\psi_N)$, on the Fibonacci chain is given by
\begin{align}
t(\psi_{i+1} + \psi_{i-1} ) + V_i \psi_i + g|\psi_i|^2\psi_i=E\psi_i
\quad \mbox{for} \ \ i=1,2,\ldots,N, 
\label{eq:gp}
\end{align}
where the hopping integral $t$ and the coupling constant $g$ are real, 
$E$ is an ``eigenenergy"~\cite{Albanese_Communications-in-Mathematical-Physics_1988_v0, 
Albanese_Communications-in-Mathematical-Physics_1988_v1, 
Albanese_Communications-in-Mathematical-Physics_1991_v0}, 
and the on-site potential $V_i$ is given by
\begin{align}
V_i=
\begin{cases}
V_A & \text{if $i$-th symbol of $S_\ell$ is $A$,}\\
V_B & \text{if $i$-th symbol of $S_\ell$ is $B$}
\end{cases}
\nonumber
\end{align}
with real $V_A$ and $V_B$. 
We impose the Dirichlet boundary conditions, $\psi_0 \! = \! \psi_{N \! + \! 1} \! = \! 0$, 
or the periodic boundary condition, $\psi_{N \! + \! i} \! = \! \psi_i$. 
We choose the normalization of the wavefunctions as $\sum_{i=1}^N|\psi_i|^2 \! = \! 1$.

\section{Numerical analysis}
\label{sec:NA} 
Using the shooting method, we numerically solve
~\eqref{eq:gp} 
with the Dirichlet boundary condition. We choose $V_A \! = \! - \! 1$, $V_B \! = \! + \! 1$ and $t \! = \! 1$.
For fixed $E$ and $g$, we continuously vary the amplitude $\psi_1$ so as to hit $\psi_{N \! + \! 1} \! = \! 0$ at 
site $N \! + \! 1$. 
A solution so obtained does not satisfy $\sum_{i=1}^N|\psi_i|^2 \! = \! 1$ in general.
So we set $\psi_i \! = \! \lambda\psi_i'$ with $\lambda>0$
so as to satisfy $\sum_{i=1}^N|\psi_i'|^2 \! = \! 1$. Then the coupling constant 
is given by $g' \! = \! \lambda^2g$.
In the following we drop primes. 

\begin{figure}[!ht]
\centering
\includegraphics[width=\figsizeb]{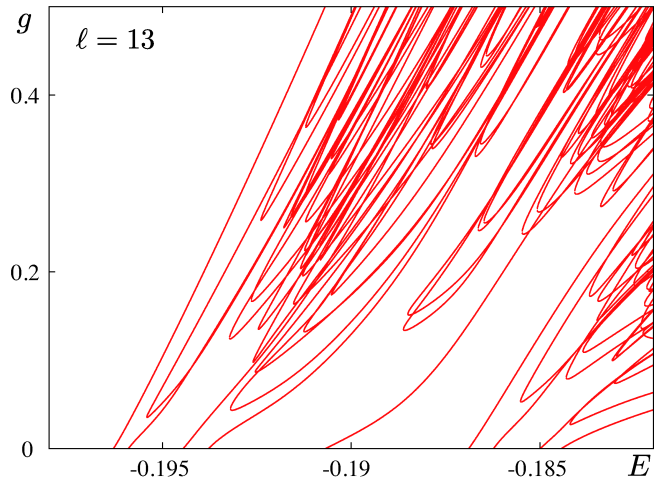}
\caption
[Energy v.s. coupling]
{
Numerical solutions of
~\eqref{eq:gp} are plotted 
in the $E$-$g$ plane near the band center 
for 
chain length $N_\ell \! = \! 377$ with $\ell \! = \! 13$.}
\label{fig:E-g}
\end{figure}

The numerical results for eigenenergies $E$ are plotted in the $E$-$g$ plane
in Fig.~\ref{fig:E-g}.
Each eigenenergy $E$ is an increasing function of $g$. 
This is a consequence of the positivity of the nonlinear term.
Unlike the linear Schr\"odinger equation, there appear many solutions
whose number is
beyond the total number of the sites of the chain $N$. 
These contain a large number of localized states.
Since the localized states are due to the nonlinearity, 
we call them solitons.
This phenomenon is already known as the appearance of 
localized states 
for lattice systems~\cite{Kivshar__2003_v0, Johansson_PRB_1994_v0}.
There are solutions with different spatial profiles
as shown in Figs.~\ref{fig:psi}.
It seems likely that all the states, extended, critical, soliton, 
and surface states, mix like a complicated ``soup".
However, we will show below that extended states are absent, 
whereas critical states and solitons exist.
The surface states exist only for a system with the Dirichlet boundary condition.

For the Harper model 
\cite{Harper_Proceedings-of-the-Physical-Society.-Section-A_1955_v0,
Aubry_Ann.-Isr.-Phys.-Soc._1980_v0} 
at criticality, effects of nonlinearity were 
discussed in \cite{Manela_New-Journal-of-Physics_2010_v0}. 
Their numerical results show that the eigenmodes can be classified into 
two families: the conventional modes which are continuously connected to those in 
the corresponding linear model and new modes that have no origin in the linear model. 
If the latter modes are interpreted as stationary solitons due to the nonlinearity as 
in our case, their results of \cite{Manela_New-Journal-of-Physics_2010_v0} are consistent with ours.  

\begin{figure}[!ht]
\centering
\includegraphics[width=\figsizec]{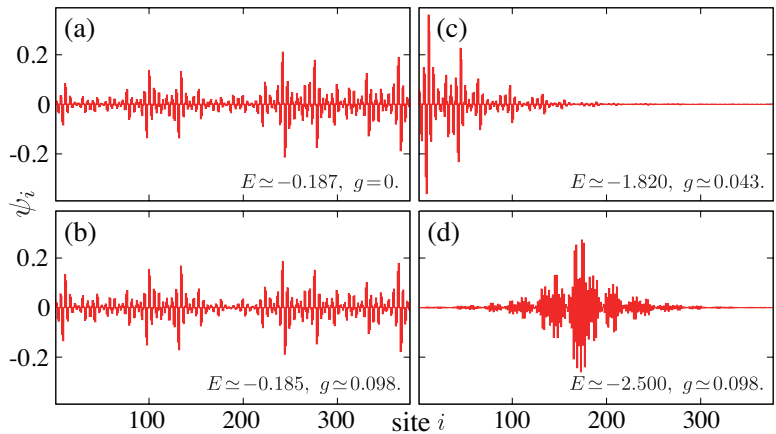}
\caption
[wavefunction]
{The left column shows 
critical wavefunctions in (a)~
a linear case and (b)~
a nonlinear case.
The right column shows 
localized wavefunctions of (c)~
a surface state and (d)~
a soliton in 
a nonlinear case. All the data are for 
chain length $N_\ell \!=\!377$ with $\ell\!=\!13$.}
\label{fig:psi}
\end{figure}

%

\section{Multifractal analysis}
\label{sec:MulA} 
As is well known, the scaling analysis which is 
called multifractal analysis~\cite{Halsey_PRA_1986_v0, Kohmoto_PRA_1988_v0} 
is a very useful tool to determine 
whether or not a given wavefunction is critical~\cite{Hiramoto_IJMPB_1992_v0}. 
Relying on this analysis, we numerically check the existence of critical states.
We expect that the wavefunction which is shown in Fig.~\ref{fig:psi}~(b) with $g\!\ne \!0$
leads to a critical state in the infinite-length limit 
of the chain. 
In the multifractal analysis, a wavefunction is characterized by $f(\alpha)$. 
For a wavefunction $\bm{\psi} \! = \! (\psi_1,\ldots,\psi_N)$,
the number of sites $i$ satisfying
$|\psi_i|^2 \! \sim \! N^{-\alpha}$ is assumed to be proportional to $N^{f_\ell(\alpha)}$. 
The subscript $\ell$ and the chain length $N$ are related through $N \! = \! N_\ell$ 
with the Fibonacci number $N_\ell$ whose definition is given in Sec.~\ref{sec:Preliminary}.  
We denote by $f (\alpha)$ the infinite-length limit of $f_\ell (\alpha)$.
Using $f(\alpha)$, one can determine the character of a given wavefunction 
as follows: 
\begin{itemize}
\item An extended state shows a single point, $(f,\alpha) \! = \! (1, 1)$. 
\item A localized state shows two points, $(f,\alpha) \! = \! (0,0)$ and $(1,\infty)$. 
\item A critical state shows a sequence of smooth curves $f_\ell(\alpha)$ 
which do not fall into the above two categories in the infinite-length limit. 
\end{itemize}
As to critical states, the corresponding sequence does not necessarily converge to 
some $f(\alpha)$. 
But, it is known that $f_\ell(\alpha)$ for a wavefunction with a multifractal character 
converges to a single smooth curve $f (\alpha)$. Our numerical results of
$f_\ell(\alpha)$ are shown in Fig.~\ref{fig:fa}.
The results show that the wavefunctions are critical 
because the behavior of the sequence of $f_\ell (\alpha)$ is totally different 
from those of $f_\ell (\alpha)$ for extended and localized states~\cite{Hiramoto_IJMPB_1992_v0}.
In consequence, the critical states persist in spite of the nonlinear term. 

However, it is still unclear whether or not the state shows
perfect multifractality~\cite{Fujiwara_PRB_1989_v0}
within our numerical analysis since the right side of the profile of $f_\ell (\alpha)$ 
is oscillating as $\ell$ increases,
and does not seem to converge to a single smooth curve in the limit $\ell\rightarrow\infty$. 
The oscillation in Fig.~\ref{fig:fa} can be explained as the effect of
the Dirichlet boundary condition as follows:
We treat only the case of the linear Schr\"odinger equation, 
i.e., \eqref{eq:gp} with $g=0$.  
We impose the periodic boundary condition. Let ${\bm\psi}\!=\!(\psi_1,\ldots,\psi_N)$ and
${\bm\psi}'\!=\!(\psi_1',\ldots,\psi_N')$ be two independent eigenvectors 
with the eigenvalues, $E$ and $E'$, respectively. We assume that $E$ is nearly equal to $E'$,
and that $|\psi_1|^2\sim N^{-\alpha_1}$ and $|\psi_1'|^2\sim N^{-\alpha_1'}$
with $\alpha_1>\alpha_1'>0$. Consider
\begin{align}
{\bm\varphi}=\bm{\psi}-\frac{\psi_1}{\psi_1'}{\bm\psi}'.
\label{phidef}
\end{align}
Then, ${\bm\varphi}$ satisfies the Dirichlet boundary condition, $\varphi_1=\varphi_{N+1}=0$.
We assume that both of ${\bm\psi}$ and ${\bm\psi}'$ show the same $f_\ell(\alpha)$. 
In order to distinguish $f_\ell(\alpha)$ of the wavefunction ${\bm\varphi}$ from 
that of $\bm{\psi}$ or ${\bm\psi}'$, we write ${\tilde f}_\ell(\alpha)$ for 
$f_\ell(\alpha)$ of ${\bm\varphi}$. 

We want to show that ${\tilde f}_\ell(\alpha)$ is not necessarily equal to $f_\ell(\alpha)$, 
and the deviation strongly depends on $\alpha_1$ and $\alpha_1'$. 
Consider the probability density $|\varphi_i|^2$ at the site $i$. 
The contribution of the first term $\bm{\psi}$ in the right-hand side of (\ref{phidef}) is 
written 
\begin{equation}
|\psi_i|^2\sim N^{-\alpha_i}
\end{equation}
with some $\alpha_i>0$. The contribution of the second term is written 
\begin{equation}
\left|\frac{\psi_1}{\psi_1'}\psi_i'\right|^2\sim 
N^{-(\alpha_1-\alpha'_1)}N^{-\alpha'_i}=N^{-(\alpha_i'+\Delta\alpha)}
\label{2ndbehav}
\end{equation}
with $\Delta\alpha:=\alpha_1-\alpha_1'$, where we have assumed 
\begin{equation}
\left|\psi_i'\right|^2\sim N^{-\alpha_i'} 
\end{equation}
with some $\alpha_i'>0$. Therefore, we naively expect that the number of the site $i$ 
satisfying $|\varphi_i|^2\sim N^{-\alpha}$ is given by 
\begin{equation}
N^{{\tilde f}_\ell(\alpha)}\sim 
\begin{cases} 
N^{f_\ell(\alpha)} & \mbox{if} \ \ f_\ell(\alpha)\ge f_\ell(\alpha-\Delta\alpha),\\
N^{f_\ell(\alpha-\Delta\alpha)} & \mbox{if} \ \ f_\ell(\alpha)<f_\ell(\alpha-\Delta\alpha).
\end{cases}
\label{Nfcases}
\end{equation}
However, the behavior in the second case is questionable. 
Clearly, from (\ref{2ndbehav}), the second case occurs for large $\alpha$. 
Further, if $\alpha=\alpha_i'+\Delta\alpha>\alpha_i$, then one has 
$|\varphi_i|^2\sim N^{-\alpha_i}$. 
In such a case, the second term in the right-hand side of (\ref{phidef})
does not contribute to ${\tilde f}_\ell(\alpha)$ although $\alpha_i'+\Delta\alpha=\alpha$. 
{From} these observations, we conclude that 
$N^{f_\ell(\alpha-\Delta\alpha)}$ in the second case of (\ref{Nfcases}) is reduced 
to some value $N^{{\tilde f}_\ell(\alpha)}$ which satisfies 
\begin{equation}
N^{f_\ell(\alpha)}\le N^{{\tilde f}_\ell(\alpha)}\le N^{f_\ell(\alpha-\Delta\alpha)}.
\end{equation}
In fact, as is well known, ${\tilde f}_\ell(\alpha)$ must become a smooth curve 
having  a single peak. To summarize, we obtain 
\begin{equation}
\begin{cases}
{\tilde f}_\ell(\alpha)=f_\ell(\alpha) & \mbox{if}\ \ f_\ell(\alpha)\ge f_\ell(\alpha-\Delta\alpha),\\ 
f_\ell(\alpha)\le {\tilde f}_\ell(\alpha)\le f_\ell(\alpha-\Delta\alpha) 
& \mbox{if}\ \ f_\ell(\alpha)<f_\ell(\alpha-\Delta\alpha).  
\end{cases}
\end{equation}
This explains the oscillation of $f_\ell(\alpha)$ on the right side of
the profile in Fig.~\ref{fig:fa}.
Thus the oscillation of the profile of the right side is due to the Dirichlet boundary condition. 

\begin{figure}[!ht]
\centering
\includegraphics[width=\figsizec]{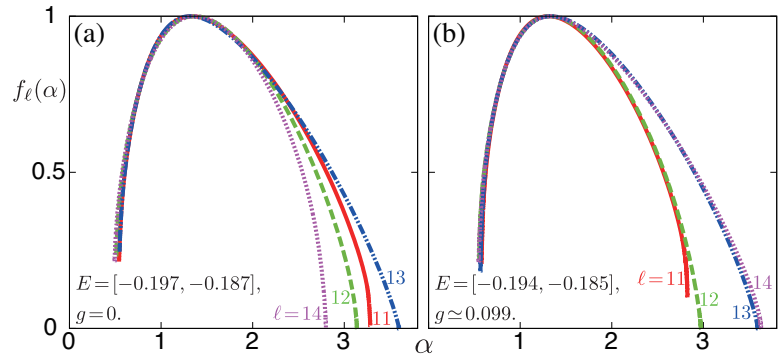}
\caption
[f-a]
{(color online)
$f_\ell(\alpha)$ in (a)~linear cases
and (b)~nonlinear cases. These numerical results show
that the corresponding states are critical.
Even in the linear case, $f_\ell(\alpha)$ oscillates 
as $\ell$ increases.}
\label{fig:fa}
\end{figure}

\section{Mathematical analysis}
\label{sec:MaA}
We want to elucidate the spectral structure of Fig.~\ref{fig:E-g}, 
and extrapolate it to the infinite-length limit of the chain (Fig.~\ref{fig:phase}). 

First of all, it is practical to review the known results in the linear case, $g \! = \! 0$. 
Consider first the linear Schr\"odinger equation 
with the periodic boundary condition. In the infinite-length limit, 
all of the stationary states are critical and the energy spectrum becomes a singular continuous 
Cantor set. 
On the other hand, as to the Dirichlet boundary condition, 
the surface states appear and their eigenenergies form a pure point spectrum 
in addition to the Cantor spectrum of the critical states. 

In this section, we prove three theorems for the nonlinear Schr\"odinger equation: 
Theorem~\ref{thm:forbidden} states that there exists a forbidden region for 
two parameters, energy and nonlinearity, 
such that the nonlinear Schr\"odinger equation~\eqref{eq:gp} with the periodic boundary condition 
has no solution. In the infinite-length limit, this forbidden region for the spectrum 
is common to both of the periodic and Dirichlet boundary conditions except for 
the spectrum of the surface states as we will show in Theorem~\ref{thm:surface}. 
Theorem~\ref{thm:zeroeffect} states that an eigenenergy of a critical or extended state 
for the nonlinear Schr\"odinger equation is included in the spectrum  
in the case of $g \! = \! 0$ in the infinite-length limit. 
This leads to the robustness of the critical states irrespective of the nonlinearity. 

Consider the \textit{linear} Schr\"odinger equation,
\begin{align}
t(\psi_{i+1} + \psi_{i-1} ) + V_i \psi_i=E\psi_i,
\label{eq:LSE}
\end{align}
i.e., 
\eqref{eq:gp} with $g \! = \! 0$
and with the periodic boundary condition. 
We denote the Hamiltonian for \eqref{eq:LSE} by $H_0$. 
We also denote the $N$ eigenvalues by $\mathcal{E}_m$, $m \! = \! 1,2,\ldots,N$, satisfying
$\mathcal{E}_j \! \le \! \mathcal{E}_k$ for $j \! < \! k$. Define a set of real numbers, 
\begin{align}
\Gamma&:=(-\infty,\mathcal{E}_1)\cup(\mathcal{E}_1 + g,\mathcal{E}_2)\cup 
\cdots\cup(\mathcal{E}_{N-1} + g,\mathcal{E}_N)\cup(\mathcal{E}_N + g, + \infty), 
\label{eq:gamma}
\end{align}
for $g \! \ge \! 0$, where $(\cdot,\cdot)$ denotes an open interval. 
Here, if $g\ge\mathcal{E}_{n \! + \! 1}-\mathcal{E}_n$, then
$(\mathcal{E}_n \! + \! g,\mathcal{E}_{n \! + \! 1})=\emptyset$ 
for $n \! = \! 1,2,\ldots,N \! - \! 1$.

\begin{thm}
Let $g>0$. Let $E$ be an eigenenergy of
the nonlinear Schr\"odinger equation~\eqref{eq:gp} 
with the periodic boundary condition. 
Then, $E \! \notin \! \Gamma$.
\label{thm:forbidden}
\end{thm}
\begin{proof}
Let $(\tilde{\psi}_1,\ldots,\tilde{\psi}_N)$ be a solution of 
the nonlinear Schr\"odinger equation~\eqref{eq:gp}
with the eigenenergy $\tilde{E}$. 
As in the assumption, we impose the periodic boundary condition.
Then one can find $n\in \! \{1, 2, \ldots, N \}$ which satisfies 
$\mathcal{E}_n \! \le \! \tilde{E} \! < \! \mathcal{E}_{n\!+\!1}$, 
where $\mathcal{E}_{N\!+\!1} \! := \! + \! \infty$. 
Therefore, it is sufficient to show that $\tilde{E} \! \le \! \mathcal{E}_n \! + \!g$.
Write $U_i \! = \! |\tilde{\psi}_i|^2$, and consider the \textit{linear} Schr\"odinger equation,
\begin{align}
t(\psi_{i+1} + \psi_{i-1})+ V_i \psi_i + g U_i \psi_i = E\psi_i,
\label{eq:LSEU}
\end{align}
with the additional potential $gU_i$, and with the periodic boundary condition. 
Namely, we fix the additional potential $gU_i$ by using the solution 
$(\tilde{\psi}_1,\ldots,\tilde{\psi}_N)$ of the nonlinear Schr\"odinger equation. 
Clearly, $(\tilde{\psi}_1,\ldots,\tilde{\psi}_N)$
is a particular solution of the equation (\ref{eq:LSEU}) with the eigenvalue $\tilde{E}$.
We denote by $\tilde{H}$ the Hamiltonian for~\eqref{eq:LSEU}, and
denote the $N$ eigenvalues by $\tilde{\mathcal{E}}_m$, $m \! = \! 1,2,\ldots,N$, satisfying
$\tilde{\mathcal{E}}_j \! \le \! \tilde{\mathcal{E}}_k$ for $j \! < \! k$.

Note that the Hamiltonian $\tilde{H}$ is written as 
$\tilde{H} \! = \! H_0 \! + \! g U$ 
with $U\!=\! \textrm{diag}(U_1,\!U_2,\!\ldots,\!U_N)$.
By the positivity of $gU$, we have 
$H_0 \! \le \! \tilde{H}$.\footnote{
When two self-adjoint operators, 
$A$ and $B$, satisfy $\left\langle \varphi, A\varphi\right\rangle
\le \left\langle \varphi,B\varphi\right\rangle$  
for any state $\varphi$, we write $A\le B$.
}
As is well known, 
by applying the min-max principle\footnote{
See, e.g., Sec.~XIII.1 
in the book~\cite{ReedSimonIV}.
} to this type of an operator 
inequality, one can get an inequality between their eigenvalues.   
In the present case, we obtain $\mathcal{E}_{n \! + \! 1} \! \le \! \tilde{\mathcal{E}}_{n \! + \! 1}$.
Combining this,
the assumption $\tilde{E} \! < \! \mathcal{E}_{n \! + \! 1}$ and the fact that $\tilde{E}$ is
the eigenvalue of $\tilde{H}$, we obtain $\tilde{E} \! \le \! \tilde{\mathcal{E}}_n$.

On the other hand, we have the bound, $\tilde{H} \! \le \! H_0 \! + \! g$. 
Applying the min-max principle again, we obtain $\tilde{\mathcal{E}}_n \! \le \! \mathcal{E}_n \! + \! g$.
Combining this with the above result $\tilde{E} \! \le \! \tilde{\mathcal{E}}_n$,
we obtain the desired bound $\tilde{E} \! \le \! \mathcal{E}_n \! + \!g$.
\end{proof}

Thus the nonlinear Schr\"odinger equation~\eqref{eq:gp} has no solution in the region 
\begin{equation}
\tilde{\Gamma} \! := \! \{(E,g)|\ E \! \in \! \Gamma,g \! \ge \! 0\}
\label{eq:gammatilde}
\end{equation}
for two parameters, energy and nonlinearity. 
The region $\tilde{\Gamma}$ is depicted as the non-colored region
in Fig.~\ref{fig:phase}. 
by replacing ``periodic" to ``Dirichlet" in the proof.
The forbidden region of the energy spectrum for the Dirichlet boundary condition 
will be treated in Theorem~\ref{thm:surface} below.
One might think that  Theorem~\ref{thm:forbidden} is not too surprising because 
the deviation of the eigenenergy is less than or equal to $g$ 
from $U_i \! = \! |\psi_i|^2 \! \le \! 1$. 
We stress that the statement of Theorem~\ref{thm:forbidden} includes 
that all the eigenenergies of the stationary solitons which are caused 
by the nonlinearity are also forbidden in the region $\tilde{\Gamma}$. 
The key idea of the proof is to introduce the potential $gU$ into the linear 
Schr\"odinger equation. This enables us to apply the min-max principle to 
nonlinear problems for the first time.

In order to determine whether a given state is critical or not, we must treat
the infinite-length limit. 
Let $\{N(k)\}_{k \! = \! 1}^\infty$ be an increasing sequence of the length $N \! = \! N(k)$ of
the present chain. 
Let $\tilde{\bm{\psi}}^{(k)} \! = \! (\tilde{\psi}_1^{(k)}, \tilde{\psi}_2^{(k)}, \ldots,
\tilde{\psi}_{N(k)}^{(k)})$ be a solution of the nonlinear Schr\"odinger equation
\eqref{eq:gp} with the eigenenergy $\tilde{E}^{(k)}$ for the chain with
the length $N(k)$. Write $U_i^{(k)} \! = \! |\tilde{\psi}_i^{(k)}|^2$, and
$U_\mathrm{max}^{(k)} \! = \! \max_i U_i^{(k)}$. We assume that
$\lim_{k\!\rightarrow\!\infty} U_\mathrm{max}^{(k)}$ exists. If not so, we take a subsequence.
We say that the wavefunction $\tilde{\bm{\psi}}^{(k)}$ is localized
if $\lim_{k\!\rightarrow\!\infty} U_\mathrm{max}^{(k)} \! > \! 0$,
and $\tilde{\bm{\psi}}^{(k)}$ is critical or extended 
if $\lim_{k\!\rightarrow\!\infty}U_\mathrm{max}^{(k)} \! = \! 0$.
We write $H_{0,N(k)}^P$ 
for $H_0$ of~\eqref{eq:LSE} with the length $N(k)$ and
with the periodic boundary condition.
The spectrum of $H_{0,N(k)}^P$ 
is given by $\sigma(H_{0,N(k)}^P) \! := \! \{\mathcal{E}_1,\ldots,\mathcal{E}_{N(k)}\}$.

\begin{thm}
Consider the nonlinear Schr\"odinger equation~\eqref{eq:gp} with periodic/Dirichlet boundary condition.
If a sequence $\{ \tilde{\bm{\psi}}^{(k)} \}_{k \! = \! 1}^\infty$
of the solutions is critical or extended in the infinite-length limit of the chain, 
$k\!\rightarrow\!\infty$, then the distance between
the eigenenergy $\tilde{E}^{(k)}$ of $\tilde{\bm{\psi}}^{(k)}$
and the spectrum $\sigma(H_{0,N(k)}^P)$ 
must go to zero in the limit $k\!\rightarrow\!\infty$.
\label{thm:zeroeffect}
\end{thm}
\begin{proof}
First consider the Hamiltonian $\tilde{H}_{N(k)}^D$ for the linear Schr\"odinger equation 
\eqref{eq:LSEU} with the Dirichlet boundary condition 
in the proof of Theorem~\ref{thm:forbidden}
with the additional potential $U_i^{(k)}$ and with the length $N(k)$. 
Clearly, $\tilde{\bm{\psi}}^{(k)}$ is the eigenvector of $\tilde{H}_{N(k)}^D$
with the eigenvalue $\tilde{E}^{(k)}$. This yields
\begin{equation}
\left \langle \tilde{\bm{\psi}}^{(k)},\left ( \tilde{H}_{N(k)}^D-H_{0,N(k)}^P\right)^2 
\tilde{\bm{\psi}}^{(k)} \right \rangle 
=\left \langle \tilde{\bm{\psi}}^{(k)},\left(\tilde{E}^{(k)}-H_{0,N(k)}^P\right)^2 
\tilde{\bm{\psi}}^{(k)}\right\rangle.
\label{E-H0P} 
\end{equation}
In order to obtain the lower bound for the right-hand side, we introduce 
the system of the complete orthonormal eigenvectors $\{\textbf{u}_i^{(k)}\}$ 
for $H_{0,N(k)}^P$. Namely, 
\begin{equation}
H_{0,N(k)}^P\textbf{u}_i^{(k)}={\cal E}_i^{(k)}\textbf{u}_i^{(k)}
\end{equation}
with the eigenvalues ${\cal E}_i~{(k)}$. Then, one has the expansion, 
\begin{equation} 
\tilde{\bm{\psi}}^{(k)}=\sum_i a_i^{(k)}\textbf{u}_i^{(k)}\quad \mbox{with} \ \ 
\sum_i \left|a_i^{(k)}\right|^2=1.
\end{equation}
Using this, the right-hand side of (\ref{E-H0P}) is evaluated as 
\begin{align}
\left \langle \tilde{\bm{\psi}}^{(k)},\left(\tilde{E}^{(k)}-H_{0,N(k)}^P\right)^2 
\tilde{\bm{\psi}}^{(k)}\right\rangle&=
\sum_i \left({\tilde E}^{(k)}-{\cal E}_i^{(k)}\right)^2\left|a_i^{(k)}\right|^2\nonumber \\ 
&\ge \min_i\left\{\left({\tilde E}^{(k)}-{\cal E}_i^{(k)}\right)^2\right\}.
\end{align}
Substituting this into the right-hand side of (\ref{E-H0P}), we obtain 
\begin{equation}
\left \langle \tilde{\bm{\psi}}^{(k)},\left ( \tilde{H}_{N(k)}^D-H_{0,N(k)}^P\right)^2 
\tilde{\bm{\psi}}^{(k)} \right \rangle
\ge \left[\mathrm{dist}(\sigma(H_{0,N(k)}^P),\tilde{E}^{(k)})\right]^2.
\label{LBvar}
\end{equation}

Next, let us obtain an upper bound for the left-hand side of~\eqref{LBvar}.
Note that $\tilde{H}_{N(k)}^D \! - \! H_{0,N(k)}^P \! = \! gU^{(k)} \! + \! \Delta H_{0,N(k)}$,
where $\Delta H_{0,N(k)} \! = \! H_{0,N(k)}^D \! - \! H_{0,N(k)}^P$, i.e., the deference between
the Hamiltonian $H_0$~\eqref{eq:LSE} with the Dirichlet and with the periodic boundary conditions.
Clearly, we have
\begin{align}
\left\langle\tilde{\bm{\psi}}^{(k)},\left(\tilde{H}_{N(k)}^D-H_{0,N(k)}^P\right)^2
\tilde{\bm{\psi}}^{(k)}\right\rangle 
&
= \left\langle\tilde{\bm{\psi}}^{(k)},g^2(U^{(k)})^2\tilde{\bm{\psi}}^{(k)}\right\rangle
+\left\langle\tilde{\bm{\psi}}^{(k)},\Delta H_{0,N(k)}^2\tilde{\bm{\psi}}^{(k)} \right\rangle
\nonumber \\
&
+ \left \langle \tilde{\bm{\psi}}^{(k)}, 
\left ( g U^{(k)} \Delta H_{0,N(k)} + \Delta H_{0,N(k)} g U^{(k)} \right ) 
\tilde{\bm{\psi}}^{(k)} 
\right \rangle. 
\label{eq:distance}
\end{align}
Since $g U^{(k)}$ and $\Delta H_{0, N(k)}$ are self-adjoint, Schwarz inequality yields
\begin{equation}
\left | \left \langle \tilde{\bm{\psi}}^{(k)}, 
g U^{(k)} \Delta H_{0,N(k)} 
\tilde{\bm{\psi}}^{(k)} 
\right \rangle \right |^2
\le \left \langle \tilde{\bm{\psi}}^{(k)}, g^2 \left (U^{(k)} \right )^2 \tilde{\bm{\psi}}^{(k)} \right \rangle 
\left \langle \tilde{\bm{\psi}}^{(k)}, \Delta H_{0,N(k)}^2 \tilde{\bm{\psi}}^{(k)} \right \rangle. 
\end{equation}
Substituting this into the right-hand side of \eqref{eq:distance}, we have 
\begin{align}
\nonumber
&\left \langle \tilde{\bm{\psi}}^{(k)},\left(\tilde{H}_{N(k)}^D-H_{0,N(k)}^P\right)^2
\tilde{\bm{\psi}}^{(k)} \right \rangle \\ \nonumber
&\le \left \langle \tilde{\bm{\psi}}^{(k)},g^2(U^{(k)})^2\tilde{\bm{\psi}}^{(k)}\right\rangle
+\left\langle\tilde{\bm{\psi}}^{(k)},\Delta H_{0,N(k)}^2\tilde{\bm{\psi}}^{(k)}\right\rangle
\nonumber \\
&+2\sqrt{\left\langle\tilde{\bm{\psi}}^{(k)},g^2(U^{(k)})^2\tilde{\bm{\psi}}^{(k)}\right\rangle
\left\langle\tilde{\bm{\psi}}^{(k)},\Delta H_{0,N(k)}^2\tilde{\bm{\psi}}^{(k)}\right\rangle}
\nonumber\\
&= \left[ \sqrt{\left\langle\tilde{\bm{\psi}}^{(k)},g^2(U^{(k)})^2\tilde{\bm{\psi}}^{(k)}\right\rangle}
+\sqrt{\left\langle\tilde{\bm{\psi}}^{(k)},\Delta H_{0,N(k)}^2\tilde{\bm{\psi}}^{(k)}\right\rangle}
\right]^2
\nonumber \\
& \le \left[|g|U_\mathrm{max}^{(k)}+|t|\sqrt{|\tilde{\psi}_1^{(k)}|^2+|\tilde{\psi}_{N(k)}^{(k)}|^2}
\right]^2.
\end{align}
Combining this with the above bound~\eqref{LBvar}, we obtain
\begin{align}
\mathrm{dist}(\sigma(H_{0,N(k)}^P),\tilde{E}^{(k)})\le
|g|U_\mathrm{max}^{(k)}+|t|\sqrt{|\tilde{\psi}_1^{(k)}|^2+|\tilde{\psi}_{N(k)}^{(k)}|^2}.
\label{distBound}
\end{align}
If $\tilde{\bm{\psi}}^{(k)}$ is critical or extended, then this right-hand side is vanishing
as $k \! \rightarrow \! \infty$.

Clearly, the same statement holds with
the periodic boundary condition. In this case, the second term in the right-hand side of
\eqref{distBound} does not appear since $\Delta H_{0, N(k)} \! = \! 0$. 
\end{proof}

We recall the well known fact that the spectrum $\sigma(H_{0,\infty}^P)$ of the linear model 
in the infinite-length limit is singular continuous and has 
zero Lebesgue measure \cite{Kohmoto_PRL_1983_v0, Kohmoto_Physics-Letters-A_1984_v0, Suto_JSP_1989_v0}. 
Theorem~\ref{thm:zeroeffect} states that all of the eigenenergies of critical or extended states 
in the nonlinear model fall into the set $\sigma(H_{0,\infty}^P)$. 
This implies that the spectrum of critical or extended states 
in the nonlinear model in the infinite-length limit is a subset of $\sigma(H_{0,\infty}^P)$, 
and has zero Lebesgue measure. In general linear models, a set of extended states is 
defined to have a spectrum having nonvanishing Lebesgue measure.  
Therefore, in this sense, there is no extended state in the present nonlinear model. 
However, we cannot conclude, from Theorem~\ref{thm:zeroeffect}, that 
critical states indeed exist in the nonlinear model, and that 
all of the critical states in the linear model survive switching on the nonlinearity.  
This expectation is supported by our numerical results. Actually, as shown in Fig.~\ref{fig:E-g}, 
each of the critical states is continuously connected to that in the linear model for 
varying the strength of the nonlinearity. 
We also remark that eigenenergies of surface states due to the Dirichlet boundary condition 
can appear outside the spectrum $\sigma(H_{0,N(k)}^P)$ in general.

The quantity $U_i=|\tilde{\psi}_i|^2$ which we introduced in the proof of Theorem~\ref{thm:forbidden}
can be interpreted as an effective potential due to the nonlinearity.  
Since $U_i$ is vanishing for critical or extended states in the infinite-length limit, 
one might think that Theorem~\ref{thm:zeroeffect} is a trivial consequence of this fact. 
This is not true because the effect of the nonlinearity for the whole chain is estimated by 
$g\sum_iU_i$. Using the normalization condition $\sum_i | \tilde{\psi}_i|^2 \! = \! 1$, 
one has $g\sum_iU_i=g\sum_i|{\tilde\psi}_i|^2=g$. 
Thus, for a fixed $g$, the effect of the nonlinearity is ${\cal O}(1)$ 
irrespective of the length of the chain.
We numerically check that 
the difference between critical wavefunctions $\bm{\psi}_{g \! = \! 0}$ 
for $g \! = \! 0$ and $\bm{\psi}_{g \! \ne \! 0}$ for $g \! \ne \! 0$ is ${\cal O} (1)$ 
in the sense of norm $||\bm{\psi}_{g\!\ne\!0} \! - \! \bm{\psi}_{g\!=\!0}||$ 
irrespective of the chain length. 
Here, $\bm{\psi}_{g \! \ne \! 0}$ is continuously connected with $\bm{\psi}_{g \! = \! 0}$. 
For a typical $\bm{\psi}_{g \! \ne \! 0}$, 
we have $|| \bm{\psi}_{g \! \ne \! 0} \! - \! \bm{\psi}_{g \! = \! 0} || \! \simeq \! 0.208$ 
for $g \! \simeq \! 0.098$ and $\ell \! = \! 13$, where $\ell$ is defined 
in Sec.~\ref{sec:Preliminary}. Here the eigenenergy 
of $\bm{\psi}_{g \! = \! 0}$ is $E \! \simeq \! -0.187$. 

Roughly speaking, nonlinearity does not change the character of eigenstates 
of a linear Schr\"odinger equation. Since the property of the on-site potential 
is not used in the proof of Theorem~\ref{thm:zeroeffect}, one can expect that 
nonlinearity does not change the localization character of eigenstates of 
a random linear Schr\"odinger equation, too. 
Actually, we can justify this type of statement for a certain class of nonlinear 
Schr\"odinger equations. The precise statement of the theorem and its proof 
are given in \ref{AppendixA}. 

We write ${\tilde \Gamma}_\infty$ for the forbidden region ${\tilde \Gamma}$ of 
\eqref{eq:gammatilde} in the infinite-length limit, i.e., 
${\tilde\Gamma}_\infty:=\lim_{N\rightarrow\infty}{\tilde \Gamma}$.
Theorem~\ref{thm:surface} below states that
the forbidden region for the spectrum of \eqref{eq:gp} 
with the Dirichlet boundary condition 
is identical to the forbidden region $\tilde{\Gamma}_\infty$ for the periodic boundary condition 
in the infinite-length limit except for the spectrum of the surface states. 
In other words, the spectrum of the surface states can appear on $\tilde{\Gamma}_\infty$.

\begin{thm}
Let $g \!\ge\! 0$. 
Let $\{\tilde{\bm{\psi}}^{(k)}\}_{k \! = \! 1}^\infty$ be a sequence of the solutions
of the nonlinear Schr\"odinger equation~\eqref{eq:gp} with the Dirichlet boundary condition and
with the eigenenergy $\tilde{E}^{(k)}$ such that
\begin{align}
\lim_{k\rightarrow\infty}\left|\tilde{\psi}_1^{(k)}\right|=
\lim_{k\rightarrow\infty}\left|\tilde{\psi}_{N(k)}^{(k)}\right|=0.
\label{zerosurface}
\end{align}
Then,
\begin{align}
\lim_{k\rightarrow\infty}\mathrm{dist}(\mathbb{R}\setminus\Gamma,\tilde{E}^{(k)})=0, 
\label{distGammacEzero}
\end{align}
where $\Gamma$ is given by \eqref{eq:gamma} and 
$\mathbb{R}\setminus\Gamma$ is the complement of $\Gamma$. 
Namely, in the infinite-length limit of the chain, 
\eqref{eq:gp} has no solution in the forbidden region $\tilde{\Gamma}_\infty$ 
except for the surface states which are localized at the surface $\{1,N\}$.
\label{thm:surface}
\end{thm}
\begin{proof}
Write $\tilde{H}_{N(k)}^P$ for the Hamiltonian $\tilde{H}$ in the proof of
Theorem~\ref{thm:forbidden} with the additional potential $U_i^{(k)} \! = \! |\tilde{\psi}_i^{(k)}|^2$
and with the chain length $N \! = \! N(k)$. Then the min-max principle yields
$\sigma(\tilde{H}_{N(k)}^P) \cap \Gamma \! = \! \emptyset$ for any $k$.
On the other hand, the same argument as in the proof of Theorem~\ref{thm:zeroeffect} yields
\begin{align}
\mathrm{dist}(\sigma(\tilde{H}_{N(k)}^P),\tilde{E}^{(k)}) \le |t|
\sqrt{|\tilde{\psi}_1^{(k)}|^2+|\tilde{\psi}_{N(k)}^{(k)}|^2}.
\end{align}
From the assumption~\eqref{zerosurface}, this right-hand side is vanishing
as $k \! \rightarrow \! \infty$. These imply the desired result~\eqref{distGammacEzero}.
\end{proof}

\begin{figure}[!ht]
\centering
\includegraphics[width=\figsizec]{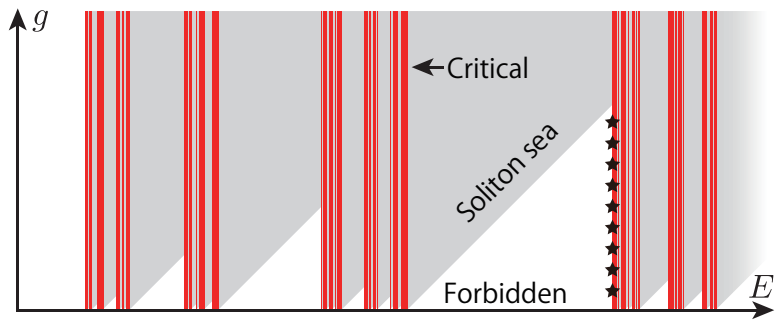}
\caption
[phase diagram]
{(color online) The spectral structure which 
consists of
the Cantor spectrum of 
critical states (vertical red lines), the sea 
of 
stationary solitons (gray),
and the forbidden region, $\tilde{\Gamma_\infty}$ (non-colored), 
in $E$-$g$ plane with the periodic boundary condition.
An example of a region where experimentally detectable critical states exist are marked as stars.}
\label{fig:phase}
\end{figure}

\section{Phase diagram of the energy spectrum}
\label{PhaseDiagram} 
Let us describe the spectral properties of
the nonlinear Schr\"odinger equation~\eqref{eq:gp} in the limit $\ell\!\rightarrow\!\infty$.
A schematic phase diagram
is shown in Fig.~\ref{fig:phase}.
It consists of three portions, the Cantor set
for the critical states, the forbidden region $\tilde{\Gamma}_\infty$, and
the soliton sea.
From the numerical results, we confirm that there exist critical states
with the finite coupling~$g$~(Fig.~\ref{fig:fa}~(b)). 
Combining this with
Theorem~\ref{thm:zeroeffect}, the Cantor set for 
the critical states is shown as vertical (red) lines in Fig.~\ref{fig:phase}.
We numerically found solitons due to the nonlinear 
effects (Fig.~\ref{fig:psi}~(d)).
The point spectrum of the solitons must be distributed outside the forbidden region $\tilde{\Gamma}_\infty$
from Theorem~\ref{thm:forbidden} and~\ref{thm:surface}.
We cannot exclude the possibility that an eigenenergy of a soliton lies just on the Cantor set.
It is obvious that surface states are absent 
with the periodic boundary condition, and not shown in Fig.~\ref{fig:phase}.

\section{Summary and conclusion}
\label{SumCon} 
We studied the stationary states for the nonlinear Schr\"odinger equation 
on the Fibonacci optical lattice. 
We found that the nonlinearity does not destroy 
the critical states 
which exist in the absence of nonlinearity and exhibit fractal properties. 
The existence of these states was confirmed numerically using multifractal analysis. 
To our knowledge, this kind of analysis is applied to the field of BEC for the first time. 
We also showed that the energy spectrum of the critical states remains 
intact irrespective of the strength of the nonlinearity. 
Besides 
the critical states, there is a large number of localized solutions, solitons, resulting 
from the nonlinearity. These solitons may seem an obstacle to observing 
critical states. 
However, we found the forbidden region for solitons, 
in the neighborhood of which the experimental detection of 
critical states is expected to be possible. 
Our analysis is intended to stimulate such an experimental effort 
to observe exotic critical states in optical lattices.

The nonlinear Schr\"odinger equation~\eqref{eq:gp} is nothing but the discrete Gross-Pitaevskii 
equation for BEC. Therefore, the chemical potential $\mu$ is equal to some eigenenergy $E$. 
In real experiments, the controllable parameters are the total number of the particles 
and the coupling constant $g$. 
In the nonlinear Schr\"odinger equation~\eqref{eq:gp}, these two parameters appear as 
a single parameter $g$ which is the effective coupling constant 
under a normalization condition of the wavefunctions. 
Within a mean field approximation for many-body BEC systems, 
the effective single-body state which has the lowest internal energy 
is most likely to be realized as the ground state.  
We numerically checked, for relatively small effective couplings $g$ (up to $1$), 
that the ground state is given by the eigenstate of \eqref{eq:gp} with the lowest eigenenergy $E$. 
Then, $E$ is identical to the corresponding the chemical potential $\mu$. 
When a chemical potential $\mu$ is given instead of the total number of the particles, 
we can also numerically determine the internal energy.  

The method presented in this paper equally applies to a higher  dimensional model 
in which the on-site potential in each direction is arranged by a generic quasiperiodic rule 
such as the Fibonacci rule. Actually, the eigenstates for the corresponding linear Schr\"odinger 
equation have a product form of the one-dimensional eigenstates. 
(See, for example, \cite{Ueda_PRL_1987_v0, Schwalm_PRB_1988_v0, Ashraff_PRB_1990_v0}.)
This implies that, in order to detect critical states or fractal wavefunctions on BEC,  
experimentalists do not necessarily need to stick to a one-dimensional system.
Also our method is applicable to a wide class of nonlinear Schr\"odinger equations 
for studying the properties of stationary states. 
For example, we can treat bichromatic on-site potentials~\cite{Larcher_PRA_2009_v0} 
which is considered as the Harper equation~\cite{Harper_Proceedings-of-the-Physical-Society.-Section-A_1955_v0, Aubry_Ann.-Isr.-Phys.-Soc._1980_v0}, 
two different hopping integrals arranged in the Fibonacci sequence~\cite{Kohmoto_PRB_1987_v0, Fujiwara_PRB_1989_v0}, 
and different types of nonlinearities
such as the Ablowitz-Ladik one~\cite{Ablowitz_JMP_1976_v0}.

\ack
The authors thank 
Yuta Masuyama, Takeshi Mizushima, Fumihiko Nakano, Mark Sadgrove, Hal Tasaki, and Satoshi Tojo for valuable discussions. 
This work was supported by Grant-in-Aid for Research Activity Start-up (23840034) 
and Grant-in-Aid for Young Scientists (B) (23740298).

\appendix
\section{Nonlinear effects for localization}
\label{AppendixA}

Although Theorem~\ref{thm:random} below holds for a wide class of nonlinear 
Schr\"odinger equations which have a localization regime in the spectrum of 
the corresponding linear Schr\"odinger equation, 
we consider a random nonlinear Schr\"odinger equation in one dimension 
as a concrete example. 
The nonlinear Schr\"odinger equation is given by replacing the on-site 
Fibonacci potential with a random potential in \eqref{eq:gp}. 
As is well known, all the eigenstates in the corresponding linear system 
are localized for a general class of randomness 
in one dimension.\footnote{See, e.g., the book~\cite{1DRSE}.}
As to the nonlinear eigenvalue problem, a certain set of localized 
eigenstates is proved to exist for more general 
setting~\cite{Albanese_Communications-in-Mathematical-Physics_1988_v0,
Albanese_Communications-in-Mathematical-Physics_1988_v1,
Albanese_Communications-in-Mathematical-Physics_1991_v0}.  
See also a related article~\cite{Frohlich_Journal-of-Statistical-Physics_1986_v0}.

In the following, we will prove that there is no stationary solution of the nonlinear 
Schr\"odinger equation such that the solution exhibits conventional properties of 
critical or extended states. 
Unfortunately, we cannot exclude the existence of certain pathological 
states which are not localized. We believe that such pathological states 
cannot appear for standard systems.  
Thus, the localization of the stationary states are expected to survive 
switching on the nonlinearity \cite{1DRE}.  
Clearly, our result is consistent with 
the previous results\cite{Albanese_Communications-in-Mathematical-Physics_1988_v0,
Albanese_Communications-in-Mathematical-Physics_1988_v1,
Albanese_Communications-in-Mathematical-Physics_1991_v0}. 

Let $\{\Lambda(k)\}_{k=1}^\infty$ be a sequence of finite lattices satisfying 
$$
\Lambda(1)\subset\Lambda(2)\subset \cdots \subset \Lambda(k) \subset \cdots .
$$
Let $\tilde{\bm{\psi}}^{(k)} \! = \! (\tilde{\psi}_1^{(k)}, \tilde{\psi}_2^{(k)}, \ldots,
\tilde{\psi}_{N(k)}^{(k)})$ be a stationary solution of the nonlinear Schr\"odinger equation 
on the lattice $\Lambda(k)$ with the eigenenergy $\tilde{E}^{(k)}$, 
where $N(k)$ is the number of the sites in $\Lambda(k)$. 
We choose the sequence $\{\Lambda(k)\}_{k=1}^\infty$ 
so that the eigenenergy $\tilde{E}^{(k)}$ 
converges to some value ${\tilde E}$ in the infinite-volume limit $k\rightarrow\infty$. 
Set 
\begin{equation}
\bm{\phi}^{(k)}:=(\phi_1^{(k)},\phi_2^{(k)},\ldots,\phi_{N(k)}^{(k)})
:=\frac{1}{\max_i{|\tilde{\psi}_i^{(k)}|}}\tilde{\bm{\psi}}^{(k)}.
\end{equation} 
For the wavefunction $\tilde{\bm{\psi}}^{(k)}$, we introduce the following two conditions:  
\begin{equation}
\lim_{k'\rightarrow\infty}\lim_{k\rightarrow\infty}\sum_{i\in\Lambda(k')}
\left|\phi_i^{(k)}\right|^2=\infty,
\label{1con}  
\end{equation}
and there exists a finite lattice $\Omega$ such that 
\begin{equation}
\mathop{\lim\inf}_{k\rightarrow\infty}\max_{i\in\Omega}\left|\phi_i^{(k)}\right|>0.
\label{2con} 
\end{equation}
The former condition (\ref{1con}) implies that the wavefunction $\tilde{\bm{\psi}}^{(k)}$ is 
not localized, and it does not split into two portions, a localized part and 
the rest, such that the distance between two portions becomes 
infinity in the infinite-volume limit. In fact, if   
\begin{equation}
\lim_{k'\rightarrow\infty}\lim_{k\rightarrow\infty}\sum_{i\in\Lambda(k')}
\left|\phi_i^{(k)}\right|^2<\infty,  
\end{equation}
then $\bm{\phi}^{(k)}$ converges to a localized state in the infinite-volume limit 
even for a critical or extended state ${\tilde{\bm{\psi}}}^{(k)}$. 

The latter condition (\ref{2con}) implies that the wavefunction $\bm{\phi}^{(k)}$ 
does not disappear from finite regions. 

We also consider the corresponding linear Schr\"odinger equation, 
and write $H_{0,N(k)}$ for the Hamiltonian. The statement is given in a generic form 
as follows: 

\begin{thm} 
Let $\tilde{\bm{\psi}}^{(k)}$ be a stationary solution 
of the nonlinear Schr\"odinger equation. 
Suppose that the eigenenergy ${\tilde E}$ in the infinite-volume limit 
is an interior point of the localization 
regime in the energy spectrum of the corresponding linear Schr\"odinger equation. 
Then the solution $\tilde{\bm{\psi}}^{(k)}$ cannot simultaneously satisfy the above two 
conditions (\ref{1con}) and (\ref{2con}).  
In other words, if a stationary solution is purely critical or extended in the sense of 
(\ref{1con}) and satisfies the energy condition, 
then its dominant part in the sense of the absolute value of the wavefunction 
cannot appear in any finite region. 
\label{thm:random}
\end{thm}

\begin{proof}
Assume that $\tilde{\bm{\psi}}^{(k)}$ satisfies the conditions (\ref{1con}) and (\ref{2con}). 
{From} this assumption, we will show that 
one can construct an extended or a critical state 
for the corresponding linear Schr\"odinger equation in the infinite-volume limit. 

{From} (\ref{2con}), there exist a subsequence $\{\Lambda(k_j)\}_{j=1}^\infty$ 
of $\{\Lambda(k)\}_{k=1}^\infty$ and 
a site $i_0\in\Omega$ such that $\phi_{i_0}^{(k_j)}$ converges to some 
nonzero value $\phi_{i_0}^{(\infty)}>0$ as $j\rightarrow\infty$.  
Therefore, we can obtain a wavefunction $\bm{\phi}^{(\infty)}$ 
in the infinite-volume limit by using the diagonal trick around the site $i_0$. 
Clearly, the wavefunction $\bm{\phi}^{(\infty)}$ is nonvanishing, and non-normalizable 
from the condition (\ref{1con}). 

In the same way as in \eqref{eq:LSEU}, we define by $\tilde{H}_{N(k)}$ the Hamiltonian 
with the additional potential $U$ which is determined by $\tilde{\psi}_i^{(k)}$. 
Then, we have $\tilde{H}_{N(k)}\bm{\phi}^{(k)} \! = \! \tilde{E}^{(k)}\bm{\phi}^{(k)}$. 

From these, we have 
\begin{align}
\lim_{j \! \rightarrow \! \infty} \left\langle \chi^{(j)},\tilde{H}_{N(\ell_j)}\bm{\phi}^{(\ell_j)}
\right\rangle
&=\left\langle \chi^{(\infty)}, H_{0,\infty}\bm{\phi}^{(\infty)} \right\rangle \nonumber \\
&=\tilde{E}\left\langle \chi^{(\infty)}, \bm{\phi}^{(\infty)} \right\rangle,
\label{eigenvalueEqA}  
\end{align}
where $\{\ell_j\}_{j=1}^\infty$ is a subsequence of $\{k_j\}_{j=1}^\infty$, and 
$\chi^{(j)}$ is a function which converges to a rapidly decreasing function 
$\chi^{(\infty)}$ in the limit $j \! \rightarrow \! \infty$, 
and $H_{0,\infty}$ is the Hamiltonian of the linear Schr\"odinger 
equation in the infinite-volume limit. 
This result implies that $\bm{\phi}^{(\infty)}$ is a generalized 
eigenvector\footnote{See, e.g., Sec.~4 of Chap.~1 in the book~\cite{Gelfand__1968_v0}.} 
of $H_{0,\infty}$.  
Since $\bm{\phi}^{(\infty)}$ is non-normalizable, this contradicts with the assumption 
that $\tilde{E}$ is an interior point of the localization regime. 
\end{proof}

The result (\ref{eigenvalueEqA}) implies that the vector 
$\bm{\phi}^{(\infty)}=(\cdots,\phi_{-1}^{(\infty)},\phi_0^{(\infty)},
\phi_1^{(\infty)},\cdots,\phi_i^{(\infty)},\cdots)$ locally satisfies 
the linear Schr\"odinger equation as 
\begin{equation}
t\left(\phi_{i+1}^{(\infty)}+\phi_{i-1}^{(\infty)}\right)+V_i\phi_i^{(\infty)}={\tilde E}\phi_i^{(\infty)}.
\end{equation}

If there exists a localization regime 
in the corresponding linear Schr\"odineger equation, then the statement of 
Theorem~\ref{thm:random} generally holds. 
For example, the Harper model has a localization regime for certain values 
of the coupling constant of the quasiperiodic potential.
Therefore, there is no stationary solution which is a conventional critical or extended state 
by switching on nonlinearity 
under the condition that the nonlinear potential $U$ vanishes in the infinite-volume limit. 
As a result, we can expect that the localization of the stationary states in the model 
is not destroyed by nonlinearity. 

Finally, we remark the following: Localization of stationary states may not lead to dynamical localization 
for nonlinear systems. (For related articles, see 
\cite{Bourgain_Diffusion-Bound-for-a-Nonlinear-Schrodinger-Equation_2007_v0, 
Pikovsky_PRL_2008_v0, Wang_JSP_2009_v0, Flach_PRL_2009_v0, Skokos_PRE_2009_v0, 
Fishman_Nonlinearity_2009_v0, Krivolapov_NJP_2010_v0}.) 
However, the existence~\cite{Albanese_Communications-in-Mathematical-Physics_1988_v0,
Albanese_Communications-in-Mathematical-Physics_1988_v1,
Albanese_Communications-in-Mathematical-Physics_1991_v0} of stationary localized states implies that, 
if nonlinearity dynamically destroys localization in a linear system,  
whether a wavefunction is dynamically localized or not strongly depends on 
the initial wavefunction. Actually, if an eigenstate is localized as 
a stationary solution of the nonlinear Schr\"odinger equation, 
then the time evolution of the state is localized, too.  

\section*{References}


\end{document}